\documentclass[journal]{IEEEtran}

\usepackage[colorinlistoftodos,prependcaption,textsize=tiny]{todonotes}

\usepackage{ifpdf}
\usepackage{cite}
\usepackage{graphicx}
\usepackage{amssymb}
\usepackage{amsmath,amsfonts}
\usepackage{stmaryrd} 
\usepackage{xcolor}
\usepackage{algorithmic}
\usepackage{textcomp}
\usepackage{bbm}
\usepackage{stmaryrd} 
\usepackage{amsthm}
\usepackage{mathrsfs}
\usepackage{tikz}
\usetikzlibrary{automata,arrows,positioning,calc}
\definecolor{celadon}{rgb}{0.67, 0.82, 0.59}
\definecolor{cblue}{rgb}{0.6, 0.73, 0.89}
\theoremstyle{remark}
\newtheorem{thm}{Theorem}
\newtheorem{prop}{Proposition}
\newtheorem{lem}{Lemma}
\newtheorem{defi}{Definition}
\newcommand{\norm}[1]{\left\lVert#1\right\rVert}
\begin{document}
	%
	\title{State Estimation over Worst-Case\\Erasure and Symmetric Channels with Memory}
	
	\author{\IEEEauthorblockN{Amir Saberi, Farhad Farokhi and Girish N.~Nair}
		\thanks{The authors are with the Department of Electrical and Electronic Engineering, University of Melbourne, VIC 3010, Australia (e-mails: asaberi@student.unimelb.edu.au,\{ffarokhi, gnair\}@unimelb.edu.au)} 
		\thanks{This work was supported by the Australian Research Council via Future Fellowship grant FT140100527, and by a McKenzie Postdoctoral Fellowship from the University of Melbourne.}
		\thanks{{\textcopyright} 2019. This manuscript version is made available under the CC-BY-NC-ND 4.0 license http://creativecommons.org/licenses/by-nc-nd/4.0/ }}
	
	\maketitle
	
	\begin{abstract}
		Worst-case models of erasure and symmetric channels are investigated, in which the number of channel errors occurring in each sliding window of a given length is bounded. Upper and lower bounds on their zero-error capacities are derived, with the lower bounds revealing a connection with the {\em topological entropy} of the channel dynamics. Necessary and sufficient conditions for linear state estimation with bounded estimation errors via such channels are then obtained, by extending previous results for non-stochastic memoryless channels to those with finite memory. These estimation conditions involve the topological entropies of the linear system and the channel.
	\end{abstract}
	
	\section{Introduction}
	When estimating the state of a linear system via a noisy memoryless channel, it is known that the relevant figure of merit for achieving estimation errors that are almost-surely uniformly bounded is the zero-error capacity $ C_0 $ of the channel, and not its ordinary capacity \cite{matveev2007shannon, Franceschetti2014}. 
	However, $ C_0 =0$ for common stochastic channel models, e.g. binary symmetric and binary erasure
	channels~\cite{korner1998zero}. This makes them unsuitable for modelling safety- or mission-critical applications that must respect hard guarantees at all times.
	Furthermore, $C_0$ depends only on the graph properties of the channel, not on the values of non-zero transition probabilities.
	
	These issues motivate the study of non-stochastic, or worst-case, channel models. Such models have received attention in the recent literature~\cite{nair2013nonstochastic, badr2017layered, wang2018end},
	and are useful when probabilistic information  about the channel noise is not available or when the noise itself is not random, e.g. in adversarial attacks. Using non-stochastic memoryless channels, it has been shown that $C_0$ still remains the relevant figure of merit for the purpose of obtaining uniformly bounded state estimation errors~\cite{nair2013nonstochastic}. 
	
	In this paper, non-stochastic models for erasure and symmetric channels are studied, whereby at most $d$ errors can occur in every sliding window of $n$ channel uses. Such models inherently have memory and therefore fall outside the class of channels considered in~\cite{nair2013nonstochastic, matveev2007shannon,shannon1956zero}. Upper bounds on $C_0$ for these channels (Theorems \ref{thm:nsec0f} and \ref{thm:nssc0f}) are derived by applying the dynamic programming equation in~\cite{zhao2010zero} for the zero-error {\em feedback} capacity of state-dependent channels with states available at transmitter and receiver. Novel lower bounds on $C_0$ are then derived in terms of the  {\em topological entropy} of the channel state dynamics (Theorems \ref{prop:nsel} and \ref{prop:nssl}). 
	Finally, the results of \cite{nair2013nonstochastic} on linear state estimation are extended to non-stochastic channels with {\em finite memory} (Theorem~\ref{thm:finitem}). This yields separate necessary and sufficient conditions for achieving uniformly bounded  estimation errors via such channels (Theorems \ref{thm:nseest} and \ref{thm:nssest}), in terms of the  topological entropies of the linear system and the channel. In a recent conference paper \cite{saberi2018estimation}, a worst-case {\em consecutive window} model of binary erasure channels was studied. Such a model can be made memoryless by a lifting argument, after which classical techniques can be applied. In contrast, the sliding-window models studied here are truly state-dependent, and require a different approach.

	Throughout the paper, $ q $ denotes the channel input alphabet size, logarithms are in base $q$, and coding rates and channel capacities are in symbols (or \emph{packets}) per channel use. The cardinality of a set is denoted by $ |\cdot| $. Define $ V_r^n(q):=1+{n \choose 1} (q-1) + \dots +{n \choose r}(q-1)^r=\sum_{i=0}^{r}{n \choose i}(q-1)^i$, where $ r \in \{0,1,\dots,n\} $. Let $ \mathbf{B}_l(u) $ be a $ l $-ball $ \{ v: \norm{u-v} \leq l\} $ centred at $ u $ with $ \norm{\cdot} $  denoting a norm on a finite-dimensional real vector space. 
	
	\section{Worst-Case Channel Models with Memory}
	
	The following two channels are studied in this paper.

	\begin{defi}[NSE channels] \label{def:nsechannel}
		A channel is called {\em $(n,d)$ non-stochastic sliding-window erasure~(NSE)} if each transmitted symbol is either received perfectly or erased, with at most $d$ erasures possible in every sliding window of past $ n $  transmissions. The receiver knows the locations of the erased symbols; however this information is not available to the sender.
	\end{defi}
	
	\begin{defi}[NSS channels]
		\label{def:nsschannel}	
		A channel is called {\em $(n,d)$ non-stochastic sliding-window symmetric~(NSS)} if the channel input and output alphabet sets are the same, each input can get mapped to any output symbol, and within every sliding-window of past $ n $  transmissions at most $ d $ errors (i.e. when the received symbols differ from the transmitted ones) can happen.
	\end{defi}
	
	The non-stochastic channels in Defs.~\ref{def:nsechannel} and~\ref{def:nsschannel}  generalize their stochastic counterparts, the binary erasure and symmetric channels. Here, instead of having a probability of error for every single use of channel, the number of errors that may occur over a sliding-window of length $n$ is upper bounded by a non-negative integer $ d $. The maximum error rate is then $d/n$.  Fig.~\ref{slidep} illustrates simple sliding-window erasure and symmetric channels with binary input alphabets, for the case of $ d =3$ and $ n =7$. 
	\begin{figure}[t]
		\includegraphics[width=60mm]{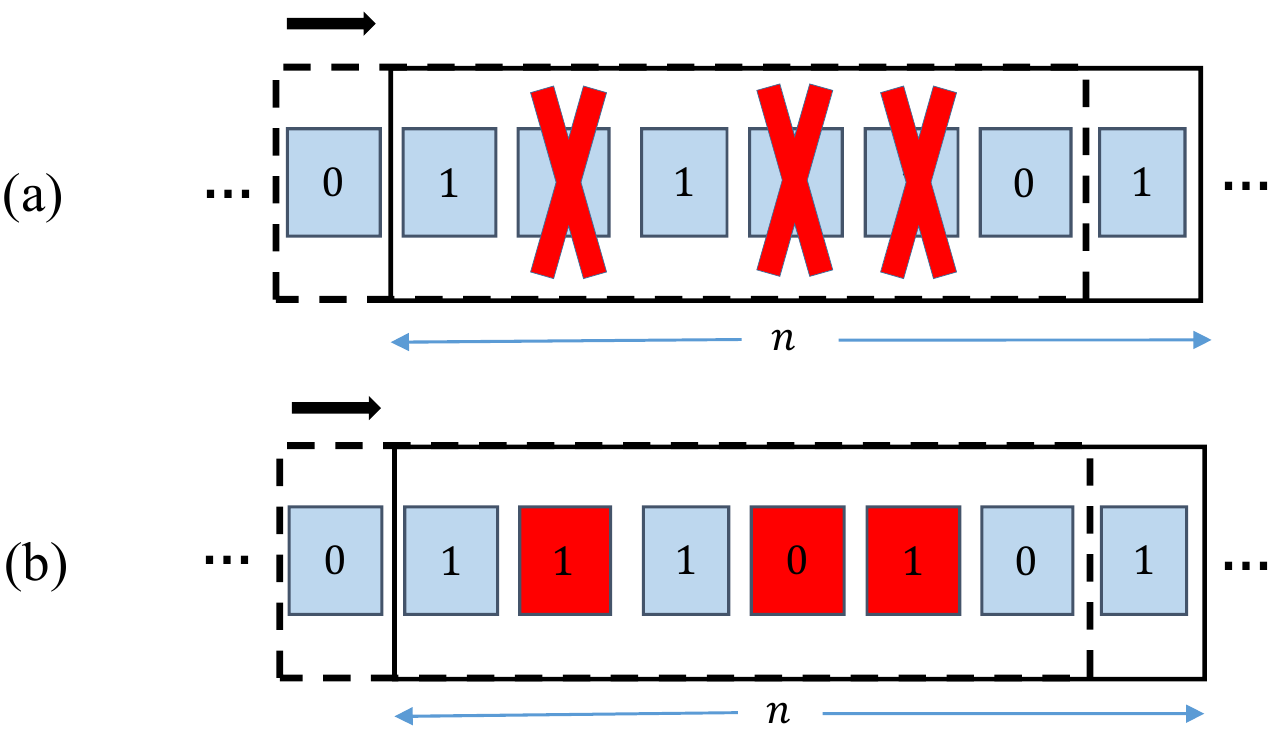}
		\centering
		\caption{Bounded error structure for binary $(7,3)$ NSE~(a) and NSS~(b) channels. }
		\label{slidep}
	\end{figure}
	The channel output depends on the errors in the previous window; thus these channels have memory.
	Equivalently, they may be represented as  \textit{state-dependent channels}, with a finite number of possible states.
	
	For an $(n,d)$ NSE channel, the current state of the channel is naturally represented as an $n$-bit word, with $*$ and $\circ$ respectively indicating the locations of erroneous and perfect transmissions in the previous window. 
	Let $\mathcal{S}$ denote the set of all possible channel states.
	For an $(n,d)$ NSE channel, combinatorial arguments easily show that  $ |\mathcal{S}|=V_d^n(2)=1+{n \choose 1} + \dots +{n \choose d}$.
	For example, a  $(3,1)$ NSE channel admits $|\mathcal{S}|=4$ states, as shown in Table~\ref{statetable}. 
	\begin{table}[t]
		\caption{States of a $(3,1)$ NSE channel}
		\centering
		\begin{tabular}{c c} 
			\hline
			\textbf{States} & Binary Representation \\  
			$ s_1 $ & $ \circ \circ \circ $  \\ 
			
			$ s_2 $ & $ \circ \circ  * $  \\
			
			$ s_3 $ & $ \circ  * \circ $  \\
			
			$ s_4 $ & $ * \circ \circ $   \\ \hline 
		\end{tabular}
		\label{statetable}
	\end{table}	
	
	Due to restrictions on the number of errors in each sliding window, not all states can be visited from any starting state in a single step. For example,  Fig.~\ref{fig:transition} shows the state transition diagram for the $(3,1)$ NSE channel.
	Let the current state of the channel be $ s_1 $. If no erasure occurs, the state of the channel remains the same; otherwise, the state transitions to $ s_2 $. In $ s_2 $, since the maximum allowed number of erasures has already occurred, the state must change to $ s_3 $. Similarly, in $ s_3 $, a  transition can occur only to $s_4$. However, in $ s_4 $, since the memory clears, another erasure may occur. Therefore, the state of the channel can transition to either $ s_1 $ or $ s_2 $. In Fig. \ref{fig:transition}, the red edges illustrate transitions in which an erasure occurs and the  black edges illustrate the error-free transitions. 
	\begin{figure}[t]
		\centering
		\begin{tikzpicture}[->, >=stealth', auto, semithick, node distance=1.7cm]
		\tikzstyle{every state}=[fill=white,draw=black,thick,text=black,scale=1]
		\node[state]   (S1)       {$  s_1 $};
		\node[state]    (S2)[left of=S1]   {$ s_2 $};
		\node[state]    (S3)[below of=S2]   {$ s_3 $};
		\node[state]    (S4)[right of=S3]   {$ s_4 $};
		\path
		(S1) edge[out=30,in=70,looseness=8] (S1)
		edge[red,bend right]   (S2)
		(S2) edge[bend right]   (S3)
		(S3) edge[bend right]   (S4)
		(S4) edge[bend right]   (S1)
		(S4) edge[red]   (S2);
		\end{tikzpicture}
		\caption{States transition diagram of a ($ d= $1,$ n= $3) NSE channel. Red line corresponds to erasure and black lines are for error-free transmission over the channel.}
		\label{fig:transition}
	\end{figure}
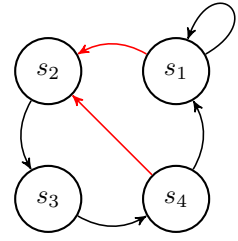
	
	For an $(n,d)$ NSS channel,  we define the state as a $q$-ary word of length $n$, in which $\circ$ indicates no error and $\circ',\circ'',\ldots$ label
	the $(q-1)$ erroneous symbol swaps that can occur.\footnote{Equivalently, by writing the channel input-output relationship as $Y_k=X_k + Z_k \mod q$, where $X_k, Y_k, Z_k \in\mathbb{Z}_q$ respectively denote the channel input, output and noise at time $k$, the current channel state is equivalent to $Z^{k-1}_{k-n}\in\mathbb{Z}_q^n$, with at most $d$ nonzero entries.}  This is not the most compact state representation; however, for a given input sequence it  yields a one-to-one relationship between the state and  output sequences, which will be useful in deriving lower bounds. The set $\mathcal{S}$ of possible states can be shown to be of size $V^n_d(q)$
	(by selecting $i=1,\ldots , d$ error locations, each with $q-1$ distinct possibilities, in a window of length $n$).
	
	To illustrate this, see Fig.~\ref{fig:trans2} for the possible states and transitions for a $(3,1)$ NSS channel with $q=3$.
	The state $s_1$ is error-free and can have $q=3$ transitions: (\textit{i}) there is no error, resulting in no change in the state of the channel, (\textit{ii}) there is an error with output $x_i+1\mod 3$, resulting in transition to state $s_2$, and (\textit{iii}) there is an error with output $x_i+2\mod 3$, resulting in transition to state $s_5$. Note that both $s_2$ and $s_5$ represent only one error; hence, in the case of NSE channel, they would have resulted in one state. Now, in state $s_2$, since it is not possible to have any more errors (because $d=1$), only one transition is possible, to state $s_3$. And so on. 
	
	\begin{figure}[t]
		\centering
		\begin{tikzpicture}[->, >=stealth', auto, semithick, node distance=1.7cm]
		\tikzstyle{every state}=[fill=white,draw=black,thick,text=black,scale=1]
		\node[state]   (S1)       {$  s_1 $};
		\node[state]    (S2)[below left of=S1]   {$ s_2 $};
		\node[state]    (S3)[left of=S1]   {$ s_3 $};
		\node[state]    (S4)[below of=S2]   {$ s_4 $};
		\node[state]    (S5)[below right of=S1]   {$ s_5 $};
		\node[state]    (S6)[right of=S1]   {$ s_6 $};
		\node[state]    (S7)[below of=S5]   {$ s_7 $};
		\path
		(S1) edge[out=70,in=110,looseness=8] (S1)
		edge[bend right,red]   (S2)
		edge[bend left,red]   (S5)
		(S2) edge[]   (S3)
		(S3) edge[bend right]   (S4)
		(S4) edge[]   (S1)
		(S4) edge[red]   (S2)
		(S4) edge[red]   (S5)
		(S7) edge[red]   (S2)
		(S5) edge[]   (S6)
		(S6) edge[bend left]   (S7)
		(S7) edge[]   (S1)
		(S7) edge[red]   (S5);
		\end{tikzpicture}
		\scalebox{.75}{\begin{tabular}{c c}
				State & Representation\\  
				\hline
				$ s_1 $ & $ \circ \circ \circ $  \\ 
				
				$ s_2 $ & $ \circ \circ  \circ' $  \\
				
				$ s_3 $ & $ \circ \circ'\circ $  \\
				
				$ s_4 $ & $ \circ'\circ \circ $   \\
				
				$ s_5 $ & $ \circ \circ \circ'' $  \\
				
				$ s_6 $ & $ \circ\circ'' \circ $  \\
				
				$ s_7 $ & $\circ'' \circ \circ $   \\\hline 
		\end{tabular}}
		\caption{States and transition diagram of $(3,1)$ NSS channel with $ q= 3$. Here $ \circ' $ and $ \circ'' $ denote possible swaps with respect to sent symbol.}
		\label{fig:trans2}
	\end{figure}
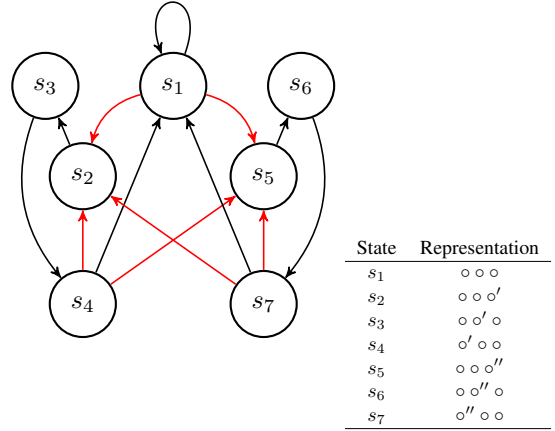
	
	Finally, the zero-error capacity $C_0$ is defined as the highest block-coding rate that permits zero decoding errors, i.e. 
	\begin{align}
	C_0:=\sup_{t \geq 0, \mathcal{F} \in \mathscr{F}}
	\frac{\log |\mathcal{F}|}{t+1}, \label{c0def}
	\end{align}
	where $\mathcal{X}$ is the input alphabet, and where $ \mathscr{F}\subseteq \mathcal{X}^{t+1}$  is the set of all block codes of length $ t+1 $ that yield zero decoding errors for any channel noise sequence and initial state.
	
	Note that for an NSE channel, $C_0$ is strictly positive if $d<n$. This is because the simple code $\{00\dots0,1\dots1,\dots, (q-1)\dots(q-1)\}$, in which every codeword is an $n$-repetition of each alphabet symbol, can send one symbol without error every $n$ transmissions, achieving a rate of $1/n$. Similarly, $C_0$ is positive for an NSS channel with $d<n/2$, since the same code achieves rate $1/n$.
	
	
	%
	
	\section{Upper Bounds using Zero-Error Feedback Capacity}
	Explicit formulas for the zero-error capacity typically do not exist except in special cases, even for memoryless channels. An upper bound on $C_0$ is the zero-error feedback capacity $C_{0f}$ with full feedback of past channel outputs back to the transmitter.  For discrete memoryless channels, $C_{0f}$ can be obtained through an optimization problem~\cite{shannon1956zero}. For state-dependent channels with causal state information at the transmitter and receiver, it has been shown that $C_{0f}$ can be obtained by solving the following sequential optimization problem~\cite{zhao2010zero}:
	\begin{align}
	C_{0f}= \liminf_{k \rightarrow \infty} \frac{1}{k} \min_{s \in \mathcal{S}} \log_q W(k,s)\label{c0Fs},
	\end{align}
	where $ \forall s \in \mathcal{S} $ and for $ k=1,2,3 , \dots $, $ W(k,s) $ is a mapping $ \mathbb{Z}^+ \times \mathcal{S} \mapsto \mathbb{R}^+ $ and is obtained iteratively (with initial value $ W(0,s)=1, \forall s \in \mathcal{S} $) from the dynamic programming (DP) equation in
	\begin{align}
	\begin{split}
	W(k,s)= \max_{P_{X|S}} &\min_{s' \in \mathcal{S}} \Bigg\{ W(k-1,s') \\ \times &\bigg(\max_{y \in \mathcal{Y}} \sum_{x \in \mathcal{G}(y,s'|s)}P_{X|S}(x|s)\bigg)^{-1} \Bigg\} \label{wit},
	\end{split}
	\end{align}
	with $ P_{X|S}(\cdot|\cdot) $ being a probability mass function on $\mathcal{X} $ for each state  $ s \in \mathcal{S} $. The subset of the inputs that can result in the output $y$ is denoted by $  \mathcal{G}(y,s'|s)=\{x| x \in \mathcal{X}, P_X(y,s'|x,s)>0\} $, in which $ s $ is the current state and $ s' $ is the next state of the channel. As an example, in NSE channels, $ \mathcal{G}(y,s'|s) =\{y\} $ if no erasure occurs and $\mathcal{G}(y,s'|s)= \mathcal{X} $ otherwise. This is because, for each transmission, each input gets uniquely mapped to an output if no erasure happens and $\mathcal{X}$ is the set of all possible inputs  if an erasure happens.
	
	In the following subsections, the zero-error feedback capacities of NSE and NSS channel models are investigated separately.
	\subsection{NSE channel}
	In the NSE channel, the state is revealed by the output sequence. Thus, in presence of an error-free feedback channel, the state is known to the encoder and the decoder, and we can apply the techniques of \cite{zhao2010zero} to yield the following formula.
	\begin{thm} \label{thm:nsec0f} The zero-error feedback capacity of an $(n,d)$ NSE channel is 
		\begin{align}
		C_{0f} =1-\frac{d}{n}.\label{nsec0f}
		\end{align}
	\end{thm}
	\begin{proof}
		See Appendix \ref{sec:appd}.
	\end{proof}
	Theorem~\ref{thm:nsec0f} states that the zero-error capacity of the NSE channel with feedback coincides with the minimum fraction of the packets that may be successfully received, given by $(n-d)/n$. 
	
	\subsection{NSS channel}
	In contrast to the NSE channel, the receiver cannot determine what errors occured from observing the output sequence. In other words, the channel states are  not known to the decoder. However, the technique of \cite{zhao2010zero} can still be used to obtain an upper bound on $ C_{0f} $, by gifting the decoder with knowledge of the states.
	
	\begin{thm} \label{thm:nssc0f} The zero-error feedback capacity of a $(n,d)$ NSS channel with $ q $-ary input alphabet is bounded by
		\begin{align}
		C_{0f} \leq 1-\frac{d}{n} \log_q(q-1),\label{nssc0f}
		\end{align}
		if $ d<n/2 $ and $ C_{0f}=0 $ when $ d \geq n/2 $.
	\end{thm}
	\begin{proof}
		See Appendix \ref{sec:appd1}.
	\end{proof}
	
	\section{Lower Bounds using Topological Entropy}\label{sec:conv}
	In this section, the  dynamics of the channel state transition diagrams are investigated, revealing a connection between zero-error capacity and the concept of {\em topological entropy} in dynamical systems theory. 
	
	Let $ s_0 $ and $ x^n=x_1x_2\dots x_n $ denote the starting state and input sequence, respectively. Define the state transition matrix $ \mathcal{A} \in \{0,1\}^{|\mathcal{S}|\times |\mathcal{S}|}$ such that the $(s,s')$th entry $\mathcal{A}_{s,s'}$ equals 1 if the state of the channel can transition from $s$ to $s'$, and equals 0 otherwise. For the case of Fig.\ref{fig:transition}, it can be seen that
	\[\mathcal{A} =\begin{bmatrix}
	1 &1 &0&0\\
	0&0&1&0\\
	0&0&0&1\\
	1&1&0&0
	\end{bmatrix}.\]
	\begin{figure}
		\centering
		\tikzstyle{level 1}=[level distance=20mm, sibling distance=45mm]
		\tikzstyle{level 2}=[level distance=20mm, sibling distance=35mm]
		\tikzstyle{level 3}=[level distance=20mm, sibling distance=20mm]
		\tikzstyle{level 4}=[level distance=15mm, sibling distance=12mm]
		\tikzstyle{level 5}=[level distance=15mm]
		\scalebox{.85}{\begin{tikzpicture}[grow=right,->,>=angle 60,scale=.8]
			\node {$s_1$}
			child {node {$s_2$}
				child {node {$s_3$}
					child {node {$s_4$}
						child {node {$ s_2 $}
							child {node {$ s_3 $}
								child {node {$s_4$}
									edge from parent
									node[above] {$a_6$}} 
								edge from parent
								node[above] {$ a_5 $}}
							edge from parent
							node[above] {$ * $}
						}
						child {node {$ s_1 $}
							child {node {$ s_2 $}
								child {node {$ s_3 $}
									edge from parent
									node[above] {$ a_6 $}}
								edge from parent
								node[above] {$ * $}
							}
							child {node {$ s_1 $}
								child {node {$ s_2 $}
									edge from parent
									node[above] {$ * $}
								}
								child {node {$ s_1 $}
									edge from parent
									node[above] {$ a_6 $}
								}
								edge from parent
								node[above] {$ a_5 $}
							}
							edge from parent
							node[above] {$ a_4 $}
						}
						edge from parent
						node[above] {$a_3$}}  
					edge from parent
					node[above] {$a_2$}
				}
				edge from parent
				node[above] {$*$}
			}
			child {node {$s_1$}
				child {node {$s_2$}
					child {node {$s_3$}
						child {node {$s_4$}
							child {node {$ s_2 $}
								child {node {$ s_3 $}
									edge from parent
									node[above] {$ a_6 $}}
								edge from parent
								node[above] {$ * $}
							}
							child {node {$ s_1 $}
								child {node {$ s_2 $}
									edge from parent
									node[above] {$ * $}
								}
								child {node {$ s_1 $}
									edge from parent
									node[above] {$ a_6 $}
								}
								edge from parent
								node[above] {$ a_5 $}
							}
							edge from parent
							node[above] {$a_4$}}  
						edge from parent
						node[above] {$a_3$}
					}
					edge from parent
					node[above] {$*$}
				}
				child {node{$s_1$}
					child {node {$s_2$}
						child {node {$s_3$}
							child {node {$s_4$}
								child{ node{$ s_2 $}
									edge from parent
									node[above] {$*$}}
								child{ node{$ s_1 $}
									edge from parent
									node[above] {$a_6$}}
								edge from parent
								node[above] {$a_5$}}  
							edge from parent
							node[above] {$a_4$}
						}
						edge from parent
						node[above] {$*$}
					}
					child {node{$ s_1 $}
						child {node {$s_2$}
							child {node {$s_3$}
								child {node {$s_4$}
									edge from parent
									node[above] {$a_6$}} 
								edge from parent
								node[above] {$a_5$}
							}
							edge from parent
							node[above] {$*$}
						}
						child {node{$s_1$}
							child {node{$ s_2 $}
								child {node {$ s_3 $}
									edge from parent
									node[above] {$ a_6 $}}
								edge from parent
								node[above] {$*$}}
							child {node{$ s_1 $}
								child {node {$ s_2 $}
									edge from parent
									node[above] {$ * $}
								}
								child {node {$ s_1 $}
									edge from parent
									node[above] {$ a_6 $}
								}
								edge from parent
								node[above] {$a_5$}}
							edge from parent
							node[above] {$a_4$}
						}
						edge from parent
						node[above] {$a_3$}}
					edge from parent
					node[above] {$a_2$}
				}
				edge from parent         
				node[above] {$a_1$}
			};
			\end{tikzpicture}}
		\caption{Possible state trajectories of $(3,1)$ NSE channel with input sequence $ x^6=a_1a_2a_3a_4a_5a_6 $.}
		\label{fig:pos}
	\end{figure}
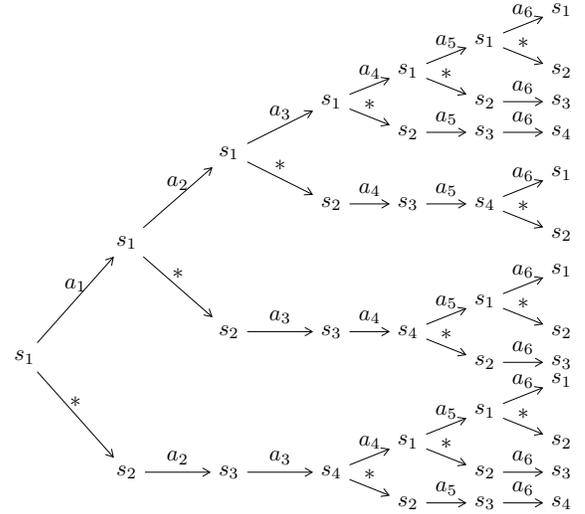
	
	Fig. \ref{fig:pos} depicts possible state transitions and related output sequences for the example channel in Fig.~\ref{fig:transition}.

	In symbolic dynamics, topological entropy  is defined as the asymptotic growth rate of the number of possible  state sequences. For a finite-state machine with an irreducible transition matrix $\mathcal{A}$, the topological entropy (in base $q$) is known to coincide with  $\log\lambda_{PF}$ , where $\lambda_{PF}$ is the Perron-Frobenius eigenvalue of $\mathcal{A}$~\cite{lind1995introduction}. This is essentially due to the fact that the number of the paths from state $ s_i $ to $ s_j $ in $ N$ steps is the $ (i,j) $-th element of $ \mathcal{A}^N$, which grows like $\lambda_{PF}^N$ for large $N$.

	For a given initial state $s_0\in\mathcal{S}$, define the binary indicator vector $ z_0\in\{0,1\}^{|\mathcal{S}|} $ consisting of all zeros except for a 1 in the position corresponding to $s_0$; e.g. in Fig.\ref{fig:transition}, if starting from state $ s_1 $, then $z_0 =[1,0,0,0]$. Let $ \mathscr{Y}(s_0,x^N) $ denote the set of all the output sequences that can occur by  transmitting the input sequence $ x^N $ from initial channel state $ s_0 $. Observe that since each output of an NSE channel (which can be a correctly received symbol or an erasure) triggers a different state transition, each sequence of state transitions has a one-to-one correspondence to the output sequence,  given the input sequence. 
	
	Based on these observations, we have the following result:
	
	
	\begin{prop}\label{prop:out}
		For any finite-state channel with an irreducible transition matrix $\mathcal{A}$, there is a positive constant
		$ \beta $ such that 
		\begin{align}
		|\mathscr{Y}(s_0,x^N) |=z_0 \mathcal{A}^N \mathbbm{1} \leq   \beta  \lambda_{PF}^N , \label{outlam}
		\end{align}
		where $ \mathbbm{1} $ is a vector of ones with appropriate dimension.
	\end{prop}
	
	\begin{proof}
		See Appendix \ref{app:pflam}.
	\end{proof}
	We now relate the zero-error capacity of the channel to its topological entropy.
	
	\begin{thm}[NSE bound via topological entropy]\label{prop:nsel}
		The zero-error capacity of an $(n,d)$ NSE channel with topological entropy $h_{ch}$ is lower-bounded by
		\begin{align}
		C_0 &\geq 1-\frac{d}{n}- h_{ch} \label{nselamd}
		\end{align}
	\end{thm}
	\begin{proof}
		See Appendix \ref{app:nselw}.
	\end{proof}
	
	\textbf{Remarks:} The topological entropy $h_{ch}$ can be viewed as the rate at which the channel dynamics generate uncertainty. Intuitively, this uncertainty cannot increase the zero-error capacity of the channel, which explains why it appears as a negative term on the RHS.
	
	There are various results that bound $ h_{ch} = \log_q\lambda_{PF}$. For instance, for any graph with maximum degree $ d_{max} $ and average degree $ d_{ave} $, we have $d_{ave} \leq \lambda_{PF} \leq d_{max}$~\cite{godsil2001strongly}. Therefore, a loose lower bound would be
	$1-d/n-\log_q d_{max}$. Moreover, note that $ d_{max}=2 $ for the state diagram of any NSE channel. Thus for large  alphabet size $q\to\infty$, the lower bound meets the upper bound obtained in \eqref{nsec0f}, i.e. $C_0\to 1-d/n$.
	
	For the example channel of Fig. \ref{fig:transition} with binary input ( $q=2$), the lower bound can be calculated to be 0.1152.

	For NSS channels, we have the following bound:
	\begin{thm}[NSS bound via topological entropy]\label{prop:nssl}
		The zero-error capacity of an $(n,d)$ NSS channel with topological entropy $h_{ch}$ bounded by
		\begin{align}
		C_0 \geq 1-2h_{ch} 
		\label{nsslmd}
		\end{align}
	\end{thm}
	\begin{proof}
		See Appendix \ref{app:nsslw}.
	\end{proof}

	\section{State Estimation over Non-stochastic Channels}\label{sec:est}
	In this section, we first briefly provide some necessary aspects of the {\em uncertain variable (uv)} framework of \cite{nair2013nonstochastic}. Using this framework, a necessary and sufficient condition for linear state estimation with uniformly bounded estimation errors via channels with finite memory is derived, extending the memoryless channel analysis in \cite{nair2013nonstochastic}.  By combining this condition with the $C_0$ bounds in previous sections, separate necessary and sufficient conditions are obtained for linear state estimation via NES and NSS channels, involving the topological entropies of the linear system and the channel.
	
	\subsection{Uncertain channels with finite memory}
	First, some definitions from \cite{nair2013nonstochastic} are needed. Let $\Pi$ be a sample space. An {\em uncertain variable (uv)} $Z$ is a mapping from $ \Pi $ to a set  $\mathcal{Z}$. 
	Given other uv's $ W $ and $Z$,  the marginal, joint and conditional ranges are denoted 
	\begin{align}
	\llbracket Z\rrbracket:=& \{ Z(\pi) : \pi \in \Pi \} \subseteq \mathcal{Z}, \nonumber\\
	\llbracket Z,W\rrbracket:=& \{ (Z(\pi),W(\pi)) : \pi \in \Pi \} \subseteq \llbracket Z\rrbracket\times \llbracket W\rrbracket,\nonumber\\
	\llbracket W|z\rrbracket:=& \{ W(\pi) : Z(\pi)=z, \pi \in \Pi \}. \nonumber
	\end{align}
	The uv's $Z$ and $W$ are said to be {\em mutually unrelated} if $\llbracket Z,W\rrbracket = \llbracket Z\rrbracket \times \llbracket W\rrbracket$, i.e. if the joint range is the Cartesian product of the marginal ones.
	
	In what follows, assume that $ \mathcal{X} $, $ \mathcal{Y} $ and $ \mathcal{V} $ are the input, output, and noise spaces of the channel, respectively. Now, a channel with finite memory can defined as follows.
	
	\begin{defi} [Uncertain channel with finite memory]
		An uncertain channel with input sequence $X$ and output sequence $Y$ is said to have a finite memory if there exists an integer $m\geq 0$ such that
		\begin{align}
		\llbracket Y(t)|&x(0:t),y(0:t-1) \rrbracket=\nonumber\\
		&\llbracket Y(t)|x(t-m:t),y(t-m:t-1) \rrbracket,\, \forall t \geq m. \label{m1}
		\end{align}
		The smallest $m$ such that \eqref{m1} holds is called the {\em memory of the channel}.
	\end{defi}
	In other words, given channel inputs and past outputs dating back $m$ steps, the current output is {\em conditionally unrelated} with the inputs and outputs that are more than $m$ steps old.
	
	Note that $m=0$ corresponds to a memoryless channel (with the convention that $y(t:t-1)$ is the empty sequence).
	
	Further note that the $(n,d)$ NSE and NSS channels considered in this paper have memory $n$. This is because the sequences $x(t-n:t-1)$ and $y(t-n:t-1)$ determine the current channel state; this, when combined with the current input $x(t)$, fully determines the range of values that the current output may take. 
	
	\subsection{State estimation of LTI systems over uncertain channels}
	
	Consider a linear time-invariant (LTI) dynamical system 
	\begin{align}
	X(t+1)&=AX(t)+V(t) \in\mathbb{R}^n, \label{lti1}\\
	Y(t)&=CX(t)+W(t) \in\mathbb{R}^p,  \label{lti2}
	\end{align}
	where the uv's $ V(t)$  and $ W(t)$ represent process and measurement disturbances. Here, the goal is to keep the estimation error \emph{uniformly bounded}, i.e. $ \sup_{t\geq 0}\|\hat{X}(t)-X(t)\|$ bounded, with $ \hat{X}(t)$ denoting the state estimate based on the measurement sequence $Y(0:t)$.The following assumptions are considered:
	\begin{itemize}
		\item[A1:] The pair $ (C,A) $ is observable;
		\item[A2:] There exist uniform bounds on the initial condition $X(0)$ and the noises $ V(t)$, $W(t)$;
		\item[A3:] The initial state $X(0)$, the noise signals $V$, $W$, and the channel error patterns are {\em mutually unrelated};
		\item[A4:]  The zero-noise sequence pair $ (V,W)=(0,0) $ is valid;
		\item[A5:] $A$ has one or more eigenvalues $\lambda_i$ with magnitude greater than one.
	\end{itemize}
	The {\em topological entropy} of the system is given by 
	\[
	h_{lin}= \sum_{|\lambda_i|\geq 1} \log|\lambda_i|,
	\]
	and can be viewed as the rate at which it generates uncertainty. We have the following theorem.
	
	\begin{thm} \label{thm:finitem}
		Consider an LTI system \eqref{lti1}-\eqref{lti2} satisfying conditions A1--A5. Assume that outputs are coded and estimated via an uncertain channel with finite memory having zero-error capacity $ C_0>0$. 
		Then a coder-estimator yielding uniformly bounded estimation errors with respect to a nonempty ball $ \mathbf{B}_l(0) \subseteq \mathbb{R}^n$ of initial states exists if and only if 
		\begin{align}
		C_0 > h_{lin}. \label{estdis}
		\end{align}
	\end{thm}
	\begin{proof}
		See Appendix \ref{sec:appd2}.
	\end{proof}
	\textbf{Remarks:} Theorem \ref{thm:finitem} extends the results of \cite{nair2013nonstochastic} for memoryless channels to channels with finite memory.
	It states that uniformly reliable estimation is possible if and only if the zero-error capacity of the channel exceeds the rate at which the system generates uncertainty.
	\subsection{State estimation over non-stochastic channels}
	In the sequel, we explore the consequences of previous results in Theorems \ref{thm:nsec0f}--\ref{thm:finitem}.
	\begin{thm} [Bounded estimation errors via NSE channel]\label{thm:nseest} 
		Consider an LTI system in \eqref{lti1}-\eqref{lti2} satisfying conditions A1--A5. Assume that the measurements are coded and transmitted via an $(n,d)$ NSE channel
		with topological entropy $h_{ch}$. Then uniformly bounded estimation errors can be achieved if
		\begin{align}
		h_{lin} + h_{ch}<1-\frac{d}{n} \label{stabi3}.
		\end{align}
		Conversely, there exist sequences of process and measurement noise for which the estimation error grows unbounded if
		\begin{align}
		h_{lin}>1-\frac{d}{n}. 
		\end{align}
	\end{thm}
	
	\tikzstyle{block1} = [draw, fill=celadon, rectangle, 
	minimum height=3em, minimum width=5em]
	\tikzstyle{block2} = [draw, fill=cblue, rectangle, 
	minimum height=3em, minimum width=5em]
	\tikzstyle{input} = [coordinate]
	\tikzstyle{output} = [coordinate]
	\tikzstyle{pinstyle} = [pin edge={to-,thin,black}]
	\begin{figure}[t]
		\centering
		\scalebox{.75}{\begin{tikzpicture}[auto, node distance=2cm,>=latex']
			\node [input, name=input ] {};
			\node [block2, right of=input, pin={[pinstyle]above: Disturbance }] (plant) { Plant };
			\node [block1, right of=plant, node distance=3cm] (encoder) { Encoder };
			\node [block2, right of=encoder, node distance=3cm] (channel) {Channel };
			\node [block1, below of=channel, node distance=1.8cm] (decoder) { Decoder/Estimator} ;
			\node [output, right of =decoder] (output){};
			
			\draw [draw,->] (input) -- node [name=w]{} (plant);
			\draw [draw,->] (plant) -- node {} (encoder);
			\draw [->] (encoder) -- node {} (channel);
			\draw [->] (channel) -- node {} (decoder);
			\draw [->] (decoder) -- node  {}(output);
			\draw [->] (w) |- (decoder);
			\end{tikzpicture}}
		\caption{State estimation via a communication channel.}
		\label{est}
	\end{figure}
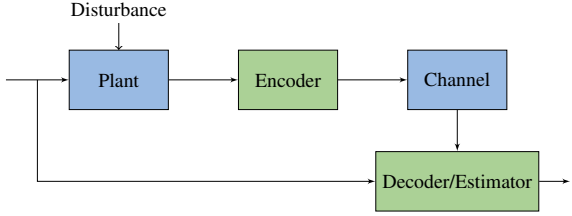
	\begin{proof} Follows from Theorems \ref{thm:nsec0f},  \ref{prop:nsel}, and \ref{thm:finitem}.  \end{proof}
	
	\textbf{Remarks:} The achievability part of this theorem involves the topological entropies of both the linear system and the channel. If their sum, which can be regarded as a total rate of uncertainty generation, is less than the worst-case rate at which symbols can be transported errorlessly across the channel, then uniformly bounded estimation errors are possible.  
	
	For non-stochastic symmetric channels, the conditions are as follows:
	
	\begin{thm} [Bounded estimation errors via NSS channel] \label{thm:nssest} Consider an LTI system in \eqref{lti1}-\eqref{lti2} satisfying conditions A1--A5. Assume that outputs are coded and estimated via a $ (n,d) $ NSS channel with topological entropy $h_{ch}$. Then, uniformly bounded estimation errors can be achieved if
		\begin{align}
		h_{lin} + 2h_{ch} <1\label{2est1}.
		\end{align}
		Conversely, there exists a sequence of process and measurement noises for which the estimation error grows unbounded if
		\begin{align}
		h_{lin}>1-\frac{d}{n}\log_q (q-1) \label{2est2}.
		\end{align}
	\end{thm}
	\begin{proof} The proof follows from Theorems \ref{thm:nssc0f}, \ref{prop:nssl}, and \ref{thm:finitem}.  \end{proof}
	%
	\section{Conclusion} \label{sec:conclusions}
	State estimation of linear time-invariant discrete-time systems over non-stochastic channels was considered. Due to the sliding nature of the channels, they had memory. Bounds for the zero-error capacity of the channels were derived using results from feedback capacity and topological entropy theory. These bounds were translated to uniformly bounded state estimation over channels with finite memory by extending the results in networked estimation theory. Future work will focus on tightening these bounds and on  the uniform stability of linear control systems via non-stochastic channels with memory. 
	
	\appendices
	
	\section{Proof of Theorem \ref{thm:nsec0f}~($ C_{0f} $ of NSE channel)} \label{sec:appd}  
	The set of states, $ \mathcal{S} $ of NSE channel can be partitioned into two subsets $ \mathcal{S}_I $ and $ \mathcal{S}_{II} $. Where, $ \mathcal{S}_I $ is the set of states that the number of erasures in past sliding-window has not reached its maximum ($ d $) and next action can take them in two states, one erasure-free transmission and the other one, with erasure. Whereas, $ \mathcal{S}_{II} $ is the set of states that there is only one possible action which corresponds to states that the number of erasure in past $n$ transmissions has reached $d$ and no erasure moves out of the window; e.g. in Fig. \ref{fig:transition}, $ \mathcal{S}_I=\{s_1,s_4\}  $ and $ \mathcal{S}_{II}=\{s_2,s_3\} $.
	Recalling that for NSE channel, $ \mathcal{G}(y,s'|s) =\{y\} $ if no erasure occurs and $\mathcal{G}(y,s'|s)= \mathcal{X} $ otherwise, the solution to first iteration of the DP problem \eqref{wit} for states in $ \mathcal{S}_I $ is
	\begin{align}
	\begin{split}
	W(1,s)&=\max_{P_{X|S}} \min\bigg\{\bigg(\sum_{x \in \mathcal{X}}P_{X|S}(x|s)\bigg)^{-1},\\
	&\qquad\qquad\qquad\qquad \bigg(\max_{y \in \mathcal{Y}} P_{X|S}(x=y|s)\bigg)^{-1}\bigg\}\end{split}\nonumber\\
	&= \min\bigg\{1,\max_{P_{X|S}}\bigg(\max_{y \in \mathcal{Y}}P_{X|S}(x=y|s)\bigg)^{-1}\bigg\}\label{its0}\\
	&= \min\{1,q\}=1,\; \forall s \in \mathcal{S}_I. \label{its1}
	\end{align}	
	where in \eqref{its0} since the summation $\sum_{x \in \mathcal{X}}P_{X|S}(x|s)$ is on all input space then it equals $1$. Furthermore, because of symmetry, uniform distribution is the solution of
	$\max_{P_{X|S}}\big(\max_{y \in \mathcal{Y}}P_{X|S}(x=y|s)\big)^{-1}$ 
	which equals $q$ and gives \eqref{its1}. 
	
	In other words, \eqref{its1} shows that when there are two possible edges out going from the initial state, the DP problem chooses the edge corresponding to erasure. 
	
	For states in $ \mathcal{S}_{II} $, there is only one edge out going which means erasure can not happen. Henceforth, there is only one state to go, thus
	\begin{align}
	W(1,s)&=\max_{P_{X|S}} \min\bigg\{\bigg(\max_{y \in \mathcal{Y}}P_{X|S}(x=y|s)\bigg)^{-1}\bigg\}\nonumber\\
	&= \max_{P_{X|S}}\bigg(\max_{y \in \mathcal{Y}}P_{X|S}(x=y|s)\bigg)^{-1}\nonumber\\
	&=q,\; \forall s \in \mathcal{S}_{II}.\label{its2}
	\end{align}
	\begin{figure}[t]
		\centering
		\begin{tikzpicture}[->, >=stealth', auto, semithick, node distance=1.7cm]
		\tikzstyle{every state}=[fill=white,draw=black,thick,text=black,scale=1]
		\node[state]   (S1)       {$  s_1 $};
		\node[state]    (S2)[left of=S1]   {$ s_2 $};
		\node[state]    (S3)[below of=S2]   {$ s_3 $};
		\node[state]    (S4)[right of=S3]   {$ s_4 $};
		\path
		(S1) edge[out=30,in=70,looseness=8] node [below right] {$ q $} (S1)
		edge[red,bend right]  node [above] {$ 1 $} (S2)
		(S2) edge[bend right]  node [left] {$ q $}  (S3)
		(S2.280) edge[blue, -,dashed, bend right, line width=.5mm]   (S3.90)
		(S3) edge[bend right]  node [below] {$ q $} (S4)
		(S3.0) edge[blue, -,dashed, bend right, line width=.5mm]   (S4.180)
		(S4) edge[bend right]  node [right] {$ q $} (S1)
		(S4) edge[red]  node [above] {$ 1 $} (S2)
		(S4.160) edge[blue, -,dashed, line width=.5mm]   (S2.290);
		\end{tikzpicture}
		\caption{The gain $a(s,s')$ associated with each edge for (3,1) NSE channel which is $q$ for error-free transmission and $1$ for erasure. The loop shaped with the states obtained from solution of DP problem for $ s \in \mathcal{S}_m $ is highlighted with blue dashed line.}
		\label{fig:transition2}
	\end{figure}
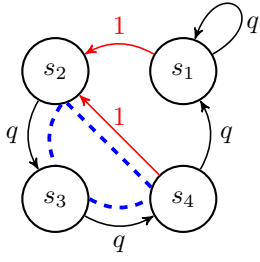
	Consequently, \eqref{its1} and \eqref{its2} show that starting from a state, DP iteration will choose erasure edge (if there is one) and the gain (or reward) $ W(1,s) =1$ and if there is no erasure edge it will go the only way possible and the gain is $ W(1,s) =q $. Observe that, the value of gains for next iterations either is $ 1 $ or $ q $. In other words, if current state $ s\in \mathcal{S}_I $, two states are reachable. Denote $ s_e $ for state that is the end-point of erasure edge and $ s_s $ for state that is end-point of error-free edge, hence
	\begin{align}
	\begin{split}
	W(k,s)&=\max_{P_{X|S}} \min\bigg\{W(k-1,s_e)\bigg(\max_{y \in \mathcal{Y}} \sum_{x \in \mathcal{X}} P_{X|S}(x|s)\bigg)^{-1}\\
	&\qquad\qquad, W(k-1,s_s)\bigg(\max_{y \in \mathcal{Y}} P_{X|S}(x=y|s)\bigg)^{-1}\bigg\}
	\end{split}\nonumber\\
	\begin{split}
	&= \min\bigg\{ W(k-1,s_e) \times 1, \\
	&\qquad \qquad \qquad W(k-1,s_s)\times q \bigg\}, \forall s \in \mathcal{S}_{I}. \label{its3}
	\end{split}
	\end{align}	
	Note that \eqref{its3} shows the edge with erasure, multiplies a gain of 1. Whereas, the edge with error-free transmission multiplies a gain of $ q $. Furthermore, if current state $ s\in \mathcal{S}_{II} $, it leads to	
	\begin{align}
	W(k,s)&= \max_{P_{X|S}} \Bigg\{W(k-1,s') \bigg(\max_{y \in \mathcal{Y}} P_{X|S}(x=y|s)\bigg)^{-1} \Bigg\}\nonumber\\
	&= W(k-1,s') \times q,\; \forall s \in \mathcal{S}_{II}. \label{its4}
	\end{align}
	
	From \eqref{its3} and \eqref{its4} it can be concluded that at each iteration a gain is multiplied to the cost-to-go, $ W(k-1,s') $. We denote this gain with $ a(s,s') $. This gain is obtained by solving $ \max_{P_{X|S}} \bigg(\max_{y \in \mathcal{Y}} \sum_{x \in \mathcal{G}(y,s'|s)} P_{X|S}(x|s)\bigg)^{-1} $ for given $ s $ and $ s' $ that results in
	\begin{align}
	a(s,s')=
	\begin{cases}
	1 & \text{if } \mathcal{G}(y,s'|s)=\mathcal{X},
	\\
	q & \text{if } \mathcal{G}(y,s'|s)=y.
	\end{cases} \label{gain}
	\end{align}
	In other words, if going from $ s $ to $ s' $ is a path or edge with erasure (or when $ \mathcal{G}(y,s'|s)= \mathcal{X}$) the gain is $a(s,s')=1 $ and if it is for error-free edge (or when $ \mathcal{G}(y,s'|s)= y$) then $a(s,s')=q $.
	In Fig. \ref{fig:transition2} the associated gain for each edge is shown. The red lines that represents erasure have a gain of $ 1 $ and other edges which represent error-free transmission have gain of $ q $.
	
	Therefore, starting from any initial state, solving the DP problem of \eqref{wit} for NSE channel, corresponds to the state trajectory with maximum number of erasures that gives minimum overall gain. This sequence for the sample channel of Fig. \ref{fig:pos} is the bottom trajectory. Now we give the following Lemma.
	
	\begin{lem}
		In the iterative solution of \eqref{wit} for NSE channel, there is a set of states $\mathcal{S}_m \subset \mathcal{S}$ such that 
		\begin{align}
		W(k,s)  = q^{n-d}W(k-n,s), \; \forall s \in \mathcal{S}_m\label{sm},
		\end{align}
	\end{lem}
	\begin{proof}
		Consider a set of states that they are associated with the situation that the number of erasures in past window of size $ n $ is reached its maximum, $ d $. The set of these state are denoted by $ \mathcal{S}_m $. Now by solving the DP problem with initial state in $ \mathcal{S}_m $, if in next action, no erasure is rolling out of window\footnote{This corresponds to situations that in the binary state representation last bit is "$*$"; e.g. "$*\circ\circ$".}, then there is only one edge directed out of the current state and have a gain of $ q $; e.g. states $ s_2 $ and $ s_3 $ in Fig. \ref{fig:transition2}. Moreover, if the current state is associated with the situation that one erasure is rolling out of the window then there are two reachable states, one corresponds to erasure (with gain $ 1 $) and another one for error-free transmission (with gain $ q $); e.g. states $ s_4 $ in Fig. \ref{fig:transition2}. Hence, because of lower gain, the edge with erasure determines the next state of DP. This procedure will continue for next iterations forming a loop that ends in starting state after $ n $ iteration. Since $ d $ number of states end up in situation with gain $ 1 $ the cumulative gain is $ q^{n-d} $.
	\end{proof}
	In other words, after $ n $ iterations, the DP problem comes back to the initial state with total gain of $ q^{n-d} $. 
	For example of Fig. \ref{fig:transition}, $ \mathcal{S}_m=\{s_2,s_3,s_4\} $ which Fig. \ref{fig:transition2} shows the loop of these states described above with a blue dashed line.		
	
	For calculating $ C_{0f} $ using \eqref{c0Fs}, we need to find the minimum $ W(\cdot,s) $ and to do so the following Lemmas are given.
	\begin{lem}\label{lem:sm} Starting from any state $ s $ in $ \mathcal{S}_m $, the following holds
		\begin{align} 
		\lim_{k \rightarrow \infty} \frac{1}{k} \log_q W(k,s)=1-\frac{d}{n}, \label{s_minf}
		\end{align}
	\end{lem}
	\begin{proof}
		From \eqref{sm}, $ \forall s \in \mathcal{S}_m $ we have
		\begin{align*}
		W(k,s)&=q^{n-d}W(k-n,s)\nonumber\\
		&=q^{l(n-d)}W(k-ln,s), \forall k =1,2,..., \lfloor k/n \rfloor.
		\end{align*}
		By choosing $ l= \lfloor k/n \rfloor$ and noting that $ 0\leq k-ln <n $, we have
		\begin{align*}
		q^{(\frac{k}{n}-1)(n-d)}W(0,s) <&  W(k,s)< q^{(\frac{k}{n})\frac{n-d}{n}}W(n,s)\\
		\Rightarrow \frac{1}{k}\log_q q^{k\frac{n-d}{n}-1} <& \frac{1}{k}\log_q W(k,s)< \frac{1}{k}\log_q q^{k\frac{n-d}{n}+1}.
		\end{align*}
		When $ k \rightarrow \infty $ the upper and lower bounds meet in $ 1-d/n $ which proves \eqref{s_minf}. 
	\end{proof}
	\begin{lem}\label{lem:s1}
		Let state $ s_1$ corresponds to the state where no erasure has happened, i.e. "$ \circ \circ \dots\circ$", we have 
		\begin{align}
		W(k,s_1)  &= \min_{s \in \mathcal{S}} W(k,s)\nonumber\\
		&= \begin{cases}
		1 & \text{if } k=1,\dots,d,
		\\
		W(k-d,s_d) & \text{if } k>d,
		\end{cases}\label{minw}
		\end{align}
		where $ s_d \in \mathcal{S}_m$ is the state that $ d $ erasures have happened in last $ d $ transmission and no erasure before; i.e. 
		\[ \underbrace{\circ  \circ \dots \circ}_{n-d} \underbrace{* \dots *}_d.\]
	\end{lem}
	\begin{proof}
		Since for next $d$ iterates DP problem goes to states with erasure which has gain 1 according to \eqref{gain}. Therefore for $ k=1,\dots ,d $, $ W(k,s_1)  = 1 $ which is the minimum possible value (no more erasure is possible). For $ k>d $ since each starting state have go to erasure states until reaching the maximum allowed erasure number. Therefore at best their value can be equal to $ W(k,s_1) $. 
	\end{proof}
	Since starting from $ s_1 $, eventually DP reaches $ \mathcal{S}_m $ and according to Lemma \ref{lem:sm} we have 
	\begin{align} 
	\lim_{k \rightarrow \infty} \frac{1}{k} \log_q W(k,s_1)&=\lim_{k \rightarrow \infty} \frac{1}{k} \log_q 1^d\times W(k-d,s_d)\nonumber\\
	&=1-\frac{d}{n},\label{s1_lim}
	\end{align}
	Lemma \ref{lem:sm}, \eqref{s1_lim} and \eqref{c0Fs} results in $ C_{0f}=1-d/n $.
	This concludes the proof of Theorem \ref{thm:nsec0f}.
	\section{Proof of Theorem \ref{thm:nssc0f}~($ C_{0f} $ of NSS channel)} \label{sec:appd1}
	
	For deriving an upper bound for the zero-error feedback capacity, we assume that the decoder has access to the state information via a {\em side channel}. Hence, \eqref{c0Fs} can be used to obtain the zero-error feedback capacity for the resultant channel. Note that here we use the same state-dependant model that is used for NSE channel at which number and position of error determines the state. In other words, with channel state information, decoder finds out that an error has happened or not such that for the non-binary channel, an error can be $ q-1 $ different symbols.
	
	For NSS channel, $ \mathcal{G}(y,s'|s)=y $ when there is no error and $ \mathcal{G}(y,s'|s)=\mathcal{X}\backslash\{y\} $ when an error occurs. Consequently, the states of NSS channel (similar to NSE channel) can be partitioned into two subsets $ \mathcal{S}_I $ and $ \mathcal{S}_{II} $. Where, $ \mathcal{S}_I $ contains states that next action can take them in two states, one error-free transmission and the other one, with error. 
	Using (11), the solution of DP problem in first iteration for states in $ \mathcal{S}_I $ is
	\begin{align}
	\begin{split}
	W(1,s)&=\max_{P_{X|S}} \min\bigg\{\bigg(\max_{y \in \mathcal{Y}} \sum_{x \in \mathcal{X}\backslash\{y\}} P_{X|S}(x|s)\bigg)^{-1},\\
	&\qquad\qquad\qquad\qquad\qquad\bigg(\max_{y \in \mathcal{Y}} P_{X|S}(x=y|s)\bigg)^{-1}\bigg\}
	\end{split}\nonumber\\
	\begin{split}
	&\leq \min\bigg\{\max_{P_{X|S}}\bigg(\max_{y \in \mathcal{Y}} \sum_{x \in \mathcal{X}\backslash\{y\}} P_{X|S}(x|s)\bigg)^{-1},\\ 
	&\qquad\qquad\qquad\qquad \max_{P_{X|S}}\bigg(\max_{y \in \mathcal{Y}} P_{X|S}(x=y|s)\bigg)^{-1}\bigg\}
	\end{split}\nonumber\\
	&= \min\bigg\{\bigg(\sum_{x \in \mathcal{X}\backslash\{y\}} \frac{1}{q}\bigg)^{-1},\bigg(\frac{1}{q}\bigg)^{-1}\bigg\}\nonumber\\
	&= \min\{\frac{q}{q-1},q\}\nonumber\\
	&=\frac{q}{q-1},\; \forall s \in \mathcal{S}_I. \label{its12}
	\end{align}
	
	Again, this shows that the DP problem choses the route with error when it is possible. For states in $ \mathcal{S}_{II} $, there is no possibility for an error to happen. Hence,
	\begin{align}
	W(1,s)&=\max_{P_{X|S}} \min\bigg\{\bigg(\max_{y \in \mathcal{Y}} P_{X|S}(x=y|s)\bigg)^{-1}\bigg\}\nonumber\\
	&= \max_{P_{X|S}}\bigg(\max_{y \in \mathcal{Y}} P_{X|S}(x=y|s)\bigg)^{-1}\nonumber\\
	&=q,\; \forall s \in \mathcal{S}_{II}.\label{its22}
	\end{align}
	Similar to NSE channel, when $ s \in \mathcal{S}_I$, we have
	\begin{align}
	\begin{split}
	W(k,s)&=\\
	\max_{P_{X|S}} &\min_{s' \in \mathcal{S}}\Bigg\{W(k-1,s') \bigg(\max_{y \in \mathcal{Y}} \sum_{x \in \mathcal{X}\backslash\{y\}}P_{X|S}(x|s)\bigg)^{-1} \Bigg\}
	\end{split}\nonumber\\
	&\leq\frac{q}{q-1} \times W(k-1,s_*) ,\; \forall s \in \mathcal{S}_{I} \label{its31}
	\end{align}
	and if next state $ s' $ is in $ \mathcal{S}_{II} $,
	\begin{align}
	W(k,s)&= \max_{P_{X|S}} \Bigg\{W(k-1,s') \bigg(\max_{y \in \mathcal{Y}}P_{X|S}(x=y|s)\bigg)^{-1} \Bigg\}\nonumber\\
	&=q \times W(k-1,s') ,\; \forall s \in \mathcal{S}_{II} \label{its41}
	\end{align}
	which we have worst-case gain as follows
	\begin{align}
	a(s,s')=
	\begin{cases}
	\frac{q}{q-1} & \text{if } \mathcal{G}(y,s'|s)=\mathcal{X}\backslash\{y\},
	\\
	q & \text{if } \mathcal{G}(y,s'|s)=y.
	\end{cases}
	\end{align}
	Similar to NSE channel, in the iterative solution for NSS channel, there is a set of states $\mathcal{S}_m \subset \mathcal{S}$ such that 
	\begin{align*}
	W(k,s)  &\leq q^{n-d}(\frac{q}{q-1})^dW(k-n,s), \; \forall s \in \mathcal{S}_m\\
	&= \frac{q^n}{(q-1)^d}W(k-n,s),
	\end{align*}
	which means that after $ n $ iterations, the DP problem comes back to starting state with gain of $ \frac{q^n}{(q-1)^d} $. 
	
	Similar to NSE channel, it is easy to see that initial state $s_1$ which corresponds to error-free state has the minimum $W(k,\cdot)$ and therefore
	\begin{align} 
	\lim_{k \rightarrow \infty} \frac{1}{k} \log_q W(k,s_1)&=\lim_{k \rightarrow \infty} \frac{1}{k} \log_q 1^d\times W(k-d,s_d)\nonumber\\
	&\leq\frac{1}{n} \log_q \bigg(\frac{q^n}{(q-1)^d}\bigg).\nonumber
	\end{align}
	Hence, since the state information is not available for original channel this gives an upper bound for zero-error feedback capacity
	\begin{align*}
	C_{0f}&\leq\lim_{k \rightarrow \infty} \frac{1}{k} \log_q W(k,s_1)\\
	&\leq \frac{1}{n} \log_q \bigg(\frac{q^n}{(q-1)^d}\bigg)\\
	&=\frac{1}{n} (\log_q q^n -\log_q (q-1)^d)\\
	&=\frac{1}{n}(n-d \log_q (q-1))\\
	&=1- \frac{d}{n}\log_q (q-1).
	\end{align*}	
	Finally, if $ d\geq n/2 $ since all inputs are adjacent $ C_{0f}=0 $. This concludes the proof.
	\section{Proof of Proposition  \ref{prop:out}~(Size of the output sequences)} \label{app:pflam}
	The total number of state trajectories after $N$-step starting from state $s_i$ is equal to sum of $i$-th row of $\mathcal{A}^N$ \cite{lind1995introduction}. Hence, because of one to one correspondence between state sequence and output sequence then $|\mathscr{Y}(s_0,x^N) |=z_0 \mathcal{A}^N \mathbbm{1}$.
	
	Next we show the upper and lower bounds of \eqref{nselamd}. According to Perron-Ferbenius Theorem, for an irreducible $k\times k$ matrix $\mathcal{A} $ (or equivalently adjacency matrix of a strongly connected graph), the entries of eigenvector $v_{PF} \in \mathbb{R}^{k}$ corresponding to $\lambda_{PF}$ are strictly positive \cite[Thm. 8.8.1]{godsil2001strongly}. Therefore, there is a transformation matrix, $T=\begin{bmatrix} v_{PF}& W \end{bmatrix} \in \mathbb{R}^{k\times k} $, (where $W \in \mathbb{R}^{k\times (k-1)}$ is the matrix consisting other eigenvectors or generalized eigenvectors) such that
	\begin{align*}
	\mathcal{A}=T\begin{bmatrix} 
	\lambda_{PF} & 0\\
	0& J
	\end{bmatrix}T^{-1},
	\end{align*}
	where $J$ is block diagonal Jordan form of remaining eigenvalues.  Assume the $i$-th element of $z_0$ is equal to 1. Hence,
	\begin{align*}
	z_0\mathcal{A}^N&=z_0\begin{bmatrix} v_{PF}& W \end{bmatrix}
	\begin{bmatrix} 
	\lambda_{PF}^N & 0\\
	0& J^N
	\end{bmatrix}T^{-1}\\
	&=\begin{bmatrix} v_{PF,i}& W_i \end{bmatrix}
	\begin{bmatrix} 
	\lambda_{PF}^N & 0\\
	0& J^N
	\end{bmatrix}
	\begin{bmatrix} T_1\\
	T_2\end{bmatrix}\\
	&=\begin{bmatrix} v_{PF,i}& W_i \end{bmatrix}
	\begin{bmatrix}\lambda_{PF}^N T_1 \\
	J^N T_2 \end{bmatrix}\\
	&=\lambda_{PF}^N ( v_{PF,i}T_1+\lambda_{PF}^{-N} W_i J^N T_2),
	\end{align*}
	where $v_{PF,i}$ is $i$-th element of $v_{PF}$,  $W_i$ is $i$-th row of $W$, and $T_1$ and $T_2$ are non-zero sub-matrices of $T^{-1}$ with appropriate dimensions. Accordingly,
	\begin{align*}
	z_0\mathcal{A}^N\mathbbm{1}&=\lambda_{PF}^N ( v_{PF,i}T_1\mathbbm{1}+\lambda_{PF}^{-N} W_i J^N T_2\mathbbm{1}),
	\end{align*}
	where $v_{PF,i}T_1\mathbbm{1}$ is a constant scalar. Moreover, since $\lambda_{PF}$ is dominant with respect to elements of $J$, the matrix $\lambda_{PF}^{-N} J^N $ is bounded by a constant matrix, $ K$ . Therefore, by choosing $\beta \geq v_{PF,i}T_1\mathbbm{1} + W_i K T_2\mathbbm{1}$ the upper bound in \eqref{nselamd} are valid. This concludes the proof.
	
	\section{Proof of Theorem  \ref{prop:nsel}~(NSE channel topological bound)} \label{app:nselw}
	In this proof, without loss of generality, assume that the channel starts from an error-free state as the channel starts enumerating errors from the first transmission. The following lemma shows why. 
	\begin{lem} \label{lemgen}
		Let $s_1$ corresponds to the error-free state then
		\begin{align*}
		\mathscr{Y}(s_i,x^N) \subseteq \mathscr{Y}(s_1,x^N), \, i=\{1,\dots, |\mathcal{S}|\}.
		\end{align*}
	\end{lem}
	\begin{proof}
		Since after $n$ transmission any state is achievable, then any sequence from $n+1$-th symbol of output sequence can be constructed from any initial state. Whereas, for first $n$ symbols of the output sequence, initiating from each state might be difference, since the number of walks are different. For initial state $s_1$ since there was no erasure in past $n$ transmissions, the least constraints are put to the upcoming erasure patterns and consequently output trajectories. Whereas, starting from other state since some erasures have happened not all combinations of erasures that are achievable from state $s_1$ can be formed in the output. Therefore, the number of walks from state $s_1$ is maximum and any combinations that are possible from other states can be constructed in the output.
	\end{proof} 
	Therefore, any zero-error code that works for initial state $s_1$ will work for other initial states as well.
	
	Let $ c^N_1 $ be the first codeword of size $ N=K n,\, K \in \mathbb{N} $ for which adjacent inputs denoted by $ \mathscr{Q}(c^N_1) $ depend on the number and position of erasures of each output sequence. Let $\mathcal{G}(y^N)$ denotes the set of input sequences that produce output sequence $ y^N $. Hence,
	\begin{align}
	\mathscr{Q}(c^N_1)&=\bigcup_{y^N \in \mathscr{Y}(s_0,c^N_1)} \mathcal{G}(y^N) \nonumber \\
	\Rightarrow |\mathscr{Q}(c^N_1)|&\leq |\mathscr{Y}(s_0,c^N_1)||\mathcal{G}(y^N)| \label{q1}
	\end{align}
	Since each path has at most $ \gamma \frac{d}{n}N=\gamma K d $ number of erasures (worst-case loop in state transition) in which $ \gamma $ is a constant. Therefore, $ |\mathcal{G}(y^N)| \leq \gamma'q^{K d} $, where $ \gamma'=q^{\gamma} $. Therefore, \eqref{q1} and \eqref{outlam} gives
	\begin{align*}
	|\mathscr{Q}(c^N_1)|&\leq \gamma'q^{K d} \times \beta \lambda_{PF}^N \\
	&=\gamma'\beta (q^{\frac{d}{n}}\lambda_{PF})^N.
	\end{align*}
	By choosing non-adjacent inputs as the codebook, results in an error-free transmission and henceforth the number of distinguishable inputs is lower bounded by
	\begin{align*}
	M(N)&\geq \frac{q^N}{|\mathscr{Q}(c^N_1)|}\\
	&\geq \frac{q^N}{\gamma'\beta (q^{\frac{d}{n}}\lambda_{PF})^N},
	\end{align*}
	Hence
	\begin{align*}
	\frac{1}{N}\log_q M(N)&\geq \log_q \frac{q}{\gamma'\beta (q^{\frac{d}{n}}\lambda_{PF})^N}\\
	&=1-\frac{d}{n}-\log_q \lambda_{PF} - \frac{1}{N} \log_q (\gamma'\beta).
	\end{align*}
	If $ N $ is large last term vanishes and result in \eqref{nselamd}.

	\section{Proof of Proposition  \ref{prop:nssl}~(NSS channel topological bounds)} \label{app:nsslw}
	First, we give the following Lemma.
	\begin{lem} \label{lemnss}
		For NSS channel, if $ y^N\in \mathscr{Y}(s_0,x^N) $ then 
		\begin{align}
		x^N \in \mathscr{Y}(s_0,y^N).
		\end{align}
	\end{lem}
	\begin{proof}
		The proof is straightforward. Let $ y^N=y_1 \dots y_N $  where some symbols in the sequence are erroneous. Now assume that $ y^N $ is fed to the channel, since the same position and number of erroneous symbols is in one of the paths of the output, the exact symbols can be flipped and $ x=x_1 \dots x_N $ can be produced.
	\end{proof}
	
	Similar to NSE channel, in this proof, without loss of generality, assume that the channel starts from an error-free state as the channel starts enumerating errors from the first transmission. The following lemma shows the reason. 
	\begin{lem} \label{lemgen2}
		Let $s_1$ corresponds to the error-free state then
		\begin{align*}
		\mathscr{Y}(s_i,x^N) \subseteq \mathscr{Y}(s_1,x^N), \, i=\{1,\dots, |\mathcal{S}|\}.
		\end{align*}
	\end{lem}
	\begin{proof}
		Initiating from state $s_1$, since there is no error in past $n$ transmissions, the least constraints are put to the upcoming error patterns. Therefore any combination of errors and symbols starting from other states can be achieved from state $s_1$, as well.  
	\end{proof}
	Similar to Appendix \ref{app:nselw}, let $ c^N_1 $ be a codeword of size $ N=K n,\, K \in \mathbb{N} $. As a result of Lemma \ref{lemnss}, $ |\mathcal{G}(y^N)|=|\mathscr{Y}(s_0,c^N_1)| $. Hence
	\begin{align*}
	|\mathscr{Q}(c^N_1)|&=|\mathscr{Y}(s_0,c^N_1)||\mathcal{G}(y^N)|\\
	&\leq \beta \lambda_{PF}^N \times \beta \lambda_{PF}^N \\
	&=\beta^2 (\lambda_{PF})^{2N}.
	\end{align*}
	Again, by choosing non-adjacent inputs as the codebook, results in an error-free transmission and accordingly
	\begin{align*}
	M(N)&\geq \frac{q^N}{|\mathscr{Q}(c^N_1)|}\\
	&\geq \frac{q^N}{\beta^2 (\lambda_{PF})^{2N}},
	\end{align*}
	Hence
	\begin{align*}
	C_0&\geq \log_q \frac{q}{\beta^2 (\lambda_{PF})^{2N}\lambda_{PF})^N}\\
	&=1-\frac{d}{n}-2\log_q \lambda_{PF} - \frac{2}{N} \log_q \beta
	\end{align*}
	If $ N $ is large last term vanishes and proves the lower bound of \eqref{nsslmd}.
	
	
	\section{Proof of Theorem \ref{thm:finitem}~(Bounded estimation over finite memory channels)} \label{sec:appd2}
	In the sequel, at first zero-error capacity for channels with memory is defined and its relationship with {\em maximin information} are derived. Next, uniformly bounded estimation problem for the LTI system \eqref{lti1}-\eqref{lti2} is given. From now on, for simplifying the representation, logarithm bases are omitted, which is still equal to size of the input alphabet set.
	\subsection{Zero-error capacity and maximin information}
	
	\begin{defi}[Taxicab Connectivity] \textcolor{white}{.}
		\begin{itemize}
			\item A pair of points $ (x,y) $ and $ (x',y') \in \llbracket X,Y \rrbracket$ are {\em taxicab connected} if a finite sequence of points $ \{(x_i,y_i)\}_{i=1}^n \subset \llbracket X,Y \rrbracket$ exists such that $ x_i=x_{i-1} $ and/or $ y_i=y_{i-1} $ for all $ i \in \{2,\dots,n\} $. Furthermore, a set $ \mathcal{A} \subset \llbracket X,Y \rrbracket $ is called taxicab connected if all points in $ \mathcal{A} $ are taxicab connected;
			\item A pair of sets $ \mathcal{A},\mathcal{B} \subset \llbracket X,Y \rrbracket $ are {\em taxicab isolated} if there are no points in $ \mathcal{A}$ and $\mathcal{B}$ that are taxicab connected;
			\item A {\em taxicab partition} of $ \llbracket X,Y \rrbracket $ is a {\em taxicab-isolated partition} $ \mathcal{T}(X;Y)=\{\mathcal{A}_i \}_{i=1}^n $ such that any $ \mathcal{A}_i $ and $ \mathcal{A}_j $ are taxicab isolated if $ i \ne j $ and $ \mathcal{A}_i $ is taxicab connected for all $ i $.
		\end{itemize}
	\end{defi}
	There exists a unique taxicab partition $ \llbracket X|Y \rrbracket_* $ that satisfies $ |\mathcal{T}(X;Y)| \leq | \llbracket X|Y \rrbracket_*| $ for any $ \llbracket X|Y \rrbracket $ \cite{nair2013nonstochastic}.
	{\em Maximin information} can be defined as 
	\begin{align}
	I_*(X;Y):=\log| \llbracket X|Y \rrbracket_* |. \label{maximin}
	\end{align}
	Based on the above definitions, we can give the following result.
	\begin{prop}
		For any uncertain channel with finite memory $ m $, 
		\begin{align}
		C_0=\lim_{t \rightarrow \infty} \sup_{X(0:t):\llbracket X \rrbracket^{t+1} }\frac{1}{t+1} I_*[X(0:t);Y(0:t)]. \label{minmaxinfo}
		\end{align}
	\end{prop}
	\begin{proof} 
		Since $ \llbracket X(0:t)|Y(0:t)\rrbracket_* $ is a partition of $ \llbracket X(0:t) \rrbracket $ which can be used as a zero-error code book. However, it depends on the initial state of the channel. But, the memory is finite and at worst case, it will effect, $ m $ symbols. Therefore,
		\begin{align}
		|\llbracket X(0:t)|Y(0:t)\rrbracket_*| \leq& \sup_{t\geq 0, \mathcal{F} \in \mathscr{F}(\mathcal{X}^{t+1})}|\mathcal{F}|+m\nonumber \\
		\Rightarrow I_*[X(0:t);Y(0:t)] -m\leq & \sup_{t\geq 0, \mathcal{F} \in \mathscr{F}(\mathcal{X}^{t+1})}  \log|\mathcal{F}| \label{i1}
		\end{align}
		Dividing \eqref{i1} by $ t+1 $, gives
		\begin{align}
		\begin{split}
		\sup_{{t\geq 0},{X(0:t):\llbracket X \rrbracket^{t+1} }}\frac{I_*[X(0:t);Y(0:t)]-m}{t+1}  \leq\\
		\sup_{{t\geq 0},{\mathcal{F} \in \mathscr{F}(\mathcal{X}(0:t))}} \frac{\log|\mathcal{F}|}{t+1} & =C_0. \label{c0l}
		\end{split}
		\end{align}
		Now, we show that $ C_0 $ can be achieved with large enough $ t $. This has been shown for memoryless uncertain channels in \cite{nair2013nonstochastic}. Omitting first $ m $ outputs for any $ t>m $ the resultant channel is memoryless from input space $ \mathcal{X}^{t+1} $ to output space $ \mathcal{Y}^{t-m+1} $, similar to \cite{feinstein1959coding} for stochastic channels. Let $ \mathscr{G}(\mathcal{X}^{t+1}) $ be the space of all zero-error codes of size $ t $ for the memoryless channel.
		Note that any coding method using this lifting is sub-additive. Let $ a(t)= \sup_{\mathcal{F} \in \mathscr{G}(\mathcal{X}^{t+1})}  |\mathcal{F}| $ then for any $ i,j>m $, $ a(i+j)\geq a(i)+a(j) $. Therefore, we have
		\begin{align}
		R^*:=&\sup_{t\in \mathbb{Z}_{\geq 0}}\frac{1}{t+1} \log a(t)\nonumber \\
		=& \lim_{t \rightarrow \infty} \frac{1}{t+1} \log a(t).
		\end{align}
		
		Let $ R_{\epsilon}=C_0-\epsilon $ be the rate slightly below the zero-error capacity that can be achieved with code length of $ t_{\epsilon}+1 $. We show that by lifting method described above, in fact rates close to $ R_{\epsilon} $ are achievable. Consider a code length of $ t=t_{\epsilon}+m+1 $ in which first $ m $ symbols are not used for data transmission and serves as a guarding space to clear the memory. Next $ t_{\epsilon}+1 $ symbols are the information that is wished to send with rate $ R_{\epsilon} $. In this case the overall rate will be as follows 
		\begin{align}
		R(t_{\epsilon},m)=&\frac{(t_{\epsilon}+1)R_{\epsilon}}{t_{\epsilon}+m+1}\nonumber\\
		=&\frac{R_{\epsilon}}{1+\frac{m}{t_{\epsilon}+1}}.
		\end{align}
		If $ t_{\epsilon} $ is large enough\footnote{We have not shown supperadditve property for the original channel, however if the rate $ R_{\epsilon} $ is achieved at code length of $ t_{\epsilon} $, it is achievable at any length of $ k t_{\epsilon}, \forall k\in \mathbb{Z}_{\geq 0} $, as well. So It is possible to make it large enough.}, Then 
		\begin{align}
		\lim_{t_{\epsilon} \rightarrow \infty} R(m,t_{\epsilon})=&\lim_{t_{\epsilon} \rightarrow \infty} R_{\epsilon}=C_0. \label{rep}
		\end{align}
		In other words, by memoryless setup described above, rates close to zero-error capacity can be achieved. On the other hand, since \eqref{minmaxinfo} holds for any memoryless channel~\cite{nair2013nonstochastic}, thus there exists a  $ \llbracket X(0:t_{\epsilon}+m)|Y(0:t_{\epsilon}+m)\rrbracket_* $ such that result in $ R_{\epsilon} $. Consequently, since  $ R_{\epsilon} $ can be made very close to $ C_0 $, the associated $ \llbracket X(0:t_{\epsilon}+m)|Y(0:t_{\epsilon}+m)\rrbracket_* $ result in $ I_*(X(0:t_{\epsilon}+m);Y(0:t_{\epsilon}+m)) $ that can be made very close to the upper bound in \eqref{c0l}. This concludes the proof.
	\end{proof}
	\subsection{State estimation of LTI systems}
	In this section, we consider the estimation of linear time invariant systems via discrete uncertain channel with finite memory. Fig. \ref{est} demonstrates the state estimation system configuration.
	
	Suppose $ S(t)=\gamma (t.Y(0,t)) \in \mathcal{S}, t\in \mathbb{Z}_{\geq 0} $ be the channel's input where $ \gamma $ is an encoder operator. Each Symbol $ S(t) $ is then transmitted over the channel with finite memory. The received symbol $ Q(t) \in \mathcal{Q} $ is decoded and a causal prediction $ \hat{X}(t+1) $ of $ X(t+1) $ is produced by means of another operator $ \eta $ as 
	\begin{align}
	\hat{X}(t+1)=\eta(t, Q(0:t)) \in \mathbb{R}^n, \,\hat{X}(0)=0.
	\end{align}
	We denote the estimation error as $ E(t):=X(t)-\hat{X}(t) $.
	
	Now, we give the proof of Theorem \ref{thm:finitem}.
	\begin{proof} {\em 1) Neccessity}:
		Assume a coder-estimator achieves uniform bounded estimation error. By change of coordinates, it can be assumed that $ A $ matrix is in {\em real Jordan canonical form} which consists of $ m $ square blocks on its diagonal, with the $ j $-th block $ A_j \in \mathbb{R}^{n_j\times n_j} ,\, j=\{1,\dots,m\}$. Let $ X_j(t),\hat{X}_j(t),E_j(t) \in \mathbb{R}^{n_j} $ and so on, be the corresponding $ j $-th component.
		
		If $ A $ has no eigenvalue with magnitude larger than $ 1 $, then the right hand side of \eqref{estdis} is zero and the inequality already holds for any capacity. Otherwise, let $ d \in \{1,\dots,n\} $ denote the number of eigenvalues with magnitude larger than $ 1 $, including repeats. From now on, we will only consider the unstable subsystem, as the stable part plays no role in the analysis. Picking
		\begin{align}
		\epsilon \in \bigg(0, 1-\max_{i:\lambda_i|>1} |\lambda_i|^{-1}\bigg), \label{eps}
		\end{align}	
		arbitrary $ \tau \in \mathbb{N} $, and dividing the interval $ [-l,l] $ on the $ i $-th axis into 
		\begin{align}
		k_i:=\lfloor|(1-\epsilon)\lambda_i|^{\tau}\rfloor, \, i \in \{1,\dots,d\} \label{kk}
		\end{align}
		equal subintervals of length $ 2l/k_i $. Let $ p_i(s),\, s=\{1,\dots,k_i \}$ denote the midpoints of the subintervals and inside each subinterval construct an interval $ \mathbf{I}_i(s) $ centered at $ p_i(s) $ with a shorter length of $ l/k_i $. A hypercuboid family is defined as below 
		\begin{align}
		\mathscr{H}&=\bigg\{\bigg( \prod_{i=1}^d \mathbf{I}_i(s_i)\bigg) :  s_i \in\{1,\dots,k_i\}, i \in \{1,\dots,d\}\bigg\}, \label{hyper}
		\end{align}
		in which any two hypercuboids are separated by a distance of $ l/k_i $ along the $ i $-th axis for each $ i \in \{1,\dots,d\}$. Now, consider an initial point with range $  \llbracket X(0)\rrbracket =\cup_{\mathbf{L}\in \mathscr{H}} \mathbf{L} \subset \mathbf{B}_l(0) \subset \mathbb{R}^d$.
		
		Let diam($ \cdot $) denote the set diameter under the $ l_{\infty} $ norm; by hypothesis, $ \exists \phi>0 $ such that
		\begin{align}
		\text{diam}  &\llbracket E_j(t)\rrbracket \geq \text{diam} \llbracket E_j(t)| q(0:t-1)\rrbracket \label{condi}\\
		=&\text{diam} \llbracket X_j(t)-\eta_j(t, q(0:t-1))| q(0:t-1)\rrbracket \nonumber\\
		=&  \text{diam} \llbracket  A^t_jX_j(0)+\sum_{i=0}^{t} A_j^{t-1-i}V_j(i)| q(0:t-1)\rrbracket \label{condi2}\\
		\geq&  \text{diam} \llbracket  A^t_jX_j(0)| q(0:t-1)\rrbracket \label{reduc}\\
		\geq& \sup_{u,v \in \llbracket  X_j(0)| q(0:t-1)\rrbracket}\frac{\norm{A^t_j(u-v)}_2}{\sqrt{n}}\nonumber\\
		\geq& \sup_{u,v \in \llbracket  X_j(0)| q(0:t-1)\rrbracket}\frac{\sigma_{min}(A^t_j)\norm{u-v}_2}{\sqrt{n}}\nonumber\\
		\geq& \sigma_{min}(A^t_j)\frac{\text{diam} \llbracket X_j(0)| q(0:t-1)\rrbracket}{\sqrt{n}}\label{diamf},
		\end{align}
		where $ \sigma_{min}(\cdot) $ denotes smallest singular value. \eqref{condi} holds since conditioning reduces the range \cite{nair2013nonstochastic}. \eqref{condi2} follows from the fact that translating does not change the range. It is reasonably easy to see that sum of two unrelated uncertain variables have larger range than their individual range which results \eqref{reduc}.
		
		Using Yamamoto identity, $ \exists t_{\epsilon} \in \mathbb{N} $ such that 
		\begin{align}
		\sigma_{min}(A^t_j) \geq (1-\frac{\epsilon}{2})^t |\lambda_{min}(A_j)|^t, \;j=\{1,\dots,m\}, t\geq t_{\epsilon}. \label{yamam}
		\end{align}
		By hypothesis, $ \exists \phi>0  $, such that 
		\begin{align}
		\phi \geq& \sup \llbracket \norm{E(t)}\rrbracket\nonumber\\
		\geq& \sup \llbracket \norm{E_j(t)}\rrbracket\nonumber\\
		\geq& 0.5 \text{diam}\sup \llbracket E_j(t)\rrbracket\nonumber\\
		\geq& \bigg((1-\frac{\epsilon}{2})|\lambda_{min}(A_j)|\bigg)^t\frac{\text{diam} \llbracket X_j(0)| q(0:t-1)\rrbracket}{2\sqrt{n}}.
		\end{align}
		Now, we show that for large enough $ \tau $, the hypercuboid family $ \mathscr{H} $	\eqref{hyper} is an $ \llbracket X(0)| q(0:\tau-1)\rrbracket $-overlap isolated partition of $ \llbracket X(0)\rrbracket $. By contradiction, suppose that $ \exists \mathbf{L}\in \mathscr{H} $ that is overlap connected in $ \llbracket X(0)| q(0:\tau-1)\rrbracket $ with another hypercuboid in $ \mathscr{H} $. Thus there exists a conditional range $ \llbracket X(0)| q(0:\tau-1)\rrbracket $ containing both a point $ u_j\in \mathbf{L} $ and a point $ v_j $ in some $ \mathbf{L}' \in \mathscr{H}\backslash \mathbf{L}$. Henceforth
		\begin{align}
		\norm{u_j-v_j}\leq& \text{diam}\llbracket X_j(0)| q(0:\tau-1)\rrbracket\nonumber\\
		\leq& \frac{2\sqrt{n}\phi}{((1-{\epsilon /2})|\lambda_{min}(A_j)|)^{\tau}} \nonumber\\
		&\;j=\{1,\dots,m\}, \tau \geq t_{\epsilon}\label{error}
		\end{align}
		Notice that, by construction any two hypercuboid in $ \mathscr{H} $ are separated by a distance of $ l/k_i $, which implies
		\begin{align*}
		\norm{u_j-v_j} &\geq \frac{l}{k_i}\\
		&=\frac{l}{\lfloor(1-\epsilon)|\lambda_i|\rfloor^{\tau}}\\
		&\geq \frac{l}{|(1-\epsilon)|\lambda_{min}(A_j)||^{\tau}}
		\end{align*}
		The right hand side of this equation would exceed the right hand side of \eqref{yamam}, when $ \tau  $ is large enough that $ \big( \frac{1-\epsilon/2}{1-\epsilon}\big)^{\tau}> 2\sqrt{n}\phi/l $, yielding a contradiction.
		
		Therefore, for sufficiently large $ \tau $, no two sets of $ \mathscr{H} $ are $ \llbracket X(0)| q(0:\tau-1)\rrbracket $-overlap connected. So,
		\begin{align}
		I_*[X(0):&Q(0:\tau-1)]=\log|\llbracket X(0)| Q(0:\tau-1)\rrbracket_*|\nonumber\\
		&\geq \log |\mathscr{H}|\nonumber\\
		&=\log \bigg(\prod_{i=1}^d k_i=\prod_{i=1}^d\lfloor |(1-\epsilon)\lambda_i|^{\tau} \rfloor\bigg)\nonumber\\	
		&\geq \log \bigg( \prod_{i=1}^d 0.5 |(1-\epsilon)\lambda_i|^{\tau}\bigg) \label{half}\\
		&=\log \bigg(2^{-d}(1-\epsilon)^{d\tau}|\prod_{i=1}^d\lambda_i|^{\tau}\bigg)\nonumber\\
		&=\tau \bigg( d\log(1-\epsilon) -\frac{d}{\tau}\log 2+\sum_{i=0}^d\log|\lambda_i| \bigg),\label{istar}
		\end{align}
		where \eqref{half} holds since $ \lfloor x \rfloor>x/2, \forall x>1$. Furthermore, condition A3 implies $ X(0) \leftrightarrow S(0 :\tau) \leftrightarrow Q(0:\tau) $. Hence,
		\begin{align*}
		I_*[X(0);Q(0:\tau-1)] &\leq I_*[X(0:\tau-1);Q(0:\tau-1)] \\
		&\leq I_*[S(0:\tau-1);Q(0:\tau-1)] \\
		&< \tau C_0 +m.
		\end{align*}
		Considering this and \eqref{istar} yields
		\begin{align*}
		C_0 > d\log(1-\epsilon) + \sum_{i=0}^d\log|\lambda_i|-\frac{d}{\tau}\log 2-\frac{m}{\tau}.
		\end{align*}
		By letting $ \tau \rightarrow \infty $ and the fact that $ \epsilon $ can be made arbitrarily small, concludes the proof of necessity.
		
		
		{\em 2) Sufficiency}: Let $ P:= \sum_{|\lambda_i|\geq 1} \log|\lambda_i|$, by \eqref{estdis} and \eqref{c0def},
		$ \forall \delta \in (0,C_0-P), \, \exists t_{\delta} >0 $ such that $ \forall \tau> t_{\delta} $, there is a zero-error code book $ \mathcal{F} \subseteq\mathcal{X}^{\tau} $, thus
		\begin{align}
		P<C_0 -\delta \leq \frac{1}{\tau} \log|\mathcal{F}|.\label{pc0}
		\end{align}
		Down-sample \eqref{lti1}-\eqref{lti2} by $ \tau $, the equivalent LTI system is
		\begin{align}
		X((k+1)\tau)&=A^{\tau}X(k\tau)+U'_{\tau}(k)+V'_{\tau}(k), \label{ltis}\\
		Y(k\tau)&=CX(k\tau)+W(k\tau), \, k\in \mathbb{Z}_{\geq0} \label{ltio}
		\end{align}
		where the accumulated control term $ U'_{\tau}(k)=\sum_{i=0}^{r} A^{\tau-1-i}BU(k\tau +i) $ and disturbance term $ V'_{\tau}(k)=\sum_{i=0}^{r} A^{\tau-1-i}V(k\tau +i) $ can be shown to be uniformly bounded over $  k \in \mathbb{Z}_{\geq0} $ for each $ r \in [0:\tau-1] $. By \eqref{pc0}, $ |\mathcal{F}| $ codewords can be transmitted for which satisfies $ \log|\mathcal{F}|> \tau P $. By the "data rate theorem" for LTI systems with bounded disturbances over error-less channels (see e.g. \cite{nair2007feedback}) then there exists a coder-estimator for the equivalent LTI system of \eqref{ltis}-\eqref{ltio} with uniformly bounded estimation error for $ k \in \mathbb{Z}_{\geq0} $.
		It readily gives the  uniformly boundedness for every $ t \in \mathbb{Z}_{\geq0} $ of the \eqref{lti1}-\eqref{lti2}.
	\end{proof}

	%
	\bibliographystyle{IEEEtran}
	\bibliography{ieeetranb}
	
\end{document}